\documentclass[11pt]{amsart}
\usepackage{latexsym}
\usepackage{graphics}
\usepackage{a4wide,amssymb}
\usepackage{amsfonts,amsmath}
\usepackage{eucal}
\usepackage{enumerate}

 \newcommand{\nn}{\nonumber}
\newcommand{\R}{{\mathbb R}}

\newcommand{\var}{\varepsilon}

\newcommand{\pa}{\partial}

\let\e=\varepsilon

\let \.=\cdot

\newtheorem{remark}{Remark}
\newtheorem{lemma}{Lemma}
\newtheorem{theorem}{Theorem}

\title{From particle systems to the Landau equation: a consistency result}
\date{\today}

\author{A. V. Boblylev}
\address{Department of Mathematics, Karlstad University, SE- 65188 Karlstad, Sweden}
\email{alexander.bobylev@kau.se}

\author{M. Pulvirenti}
\address{Dipartimento di
Matematica, Universit\`a La Sapienza - Roma, P.le A. Moro, 5 00185 Roma, Italy}
\email{pulvirenti@mat.uniroma1.it}

\author{C. Saffirio}
\address{Dipartimento di
Matematica, Universit\`a La Sapienza - Roma, P.le A. Moro, 5 00185 Roma, Italy}
\email{saffirio@mat.uniroma1.it}

\begin{document}
\maketitle

\begin{abstract}
We consider a  system of $N$ classical particles,  interacting via a smooth, short-range potential, in a weak-coupling regime.  This means that $N$ tends to infinity when the interaction is suitably rescaled. The $j$-particle marginals,  which obey to the usual BBGKY hierarchy, are decomposed into two contributions: one small but strongly oscillating, the other hopefully smooth. Eliminating the first, we arrive to establish the dynamical problem in term of a new hierarchy (for the smooth part) involving a memory term. We show that the first order correction to the free flow converges, as $N\to \infty$, to the corresponding term associated to the Landau equation. We also show the related propagation of chaos. 
\end{abstract}
%

%%%%%%%%%%%%%%%%%%%%%%%%%%%%%%%%%%%%%%%%%%%%%%%
%%%%%%%%%%%%%%%%%%%%%%%%%%%%%%%%%%%%%%%%%%%%%%%

\section{Introduction}
\label{intro}
Lev Landau in 1936 proposed a kinetic equation, usually called Fokker-Planck-Landau equation
(simply Landau equation in the sequel) which is a diffusion with friction in velocity,  suitable to describe the behavior
of a  weakly interacting  gas, in particular a Coulomb gas in a regime where
the grazing collisions are dominant.

Roughly speaking the Landau's argument was to take the Boltzmann equation with Coulomb cross-section and
(cutting-off short and long distances) apply  the Taylor expansion to the collision operator.
The result is a degenerate elliptic operator acting on  the velocity space (see \cite {LL} and  the original publication of Landau \cite{L}).
The full Taylor expansion of the Boltzmann collision integral for arbitrary intermolecular forces was studied in \cite{Bo-76} and a formal generalization of Landau collision integral to arbitrary scattering cross-section was proposed there. A more precise asymptotics in the Coulomb case was also studied in \cite{DL}.

The Landau equation for the one particle distribution $f(x,v,t)$, where $x\in\R^3$, $v\in\R^3$ and $t\in\R_+$ denote  position, velocity and time respectively, reads as
\begin{equation}
\label{landau}
(\pa_t+v\cdot \nabla_x)f=Q_L(f,f)
\end{equation}
with the collision operator $Q_L$ given by:
\begin{equation}
\label{QL}
Q_L(f,f)(v)= \int dv_1 \nabla_{v} \left[ a
(v-v_1)
\; (\nabla_{v}-
\nabla_{v_1})\;f(v)f(v_1)\right].
\end{equation}
Here $x$ plays the role of a parameter and hence its dependence is
omitted.
Moreover the matrix $a(w) $ has the form

\begin{equation}
\label{matrice_a}
a(w)=\frac A{|w|} \frac {(|w|^2 Id-w \otimes w)} {|w|^2},
\end{equation}
where $A>0$ is a suitable constant.

Note that the Landau equation possesses all the properties known for the Boltzmann equation, namely
the mass, momentum and energy conservation and the H-theorem.
Actually the   homogeneous  Landau equation
can be rigorously derived in the grazing collision limit of the homogeneous Boltzmann equation
by a suitable rescaling of the cross-section.

In particular,
in \cite{AB} the authors show that, under suitable assumptions on the
cross--section, the diffusion Landau equation
\eqref{landau} can indeed be derived. The
diffusion operator is the form  \eqref{QL} but
with a matrix $a$ replaced by
$$
\alpha(|w|)\frac{(|w|^2 Id-w \otimes w)}{|w|^2},
$$
with $\alpha$ a smooth function.
Next in \cite{Gou} and \cite{Vi1} steps forward were performed
to arrive to cover the case $\alpha(|w|) \approx \tfrac 1 {|w|^\nu}$
for small $|w|$, with $\nu<1$.

The case of the matrix \eqref {matrice_a} was treated in \cite {Vi2}.
It is worth to underline that the initial value problem
 for the homogeneous Landau equation
is strongly simplified for the case $\alpha(|w|) \approx \tfrac 1 {|w|^\nu}$,
with $\nu<1$ (see \cite{DV} and \cite {DV1}),
while for the matrix \eqref{matrice_a} we have a weak existence theorem
obtained by compactness arguments based on the entropy production control
\cite {Vi2}.  Moreover, for the inhomogeneous case, we have existence and uniqueness of strong solutions 
for data sufficiently close to a Maxwellian \cite {Guo}. This is the only existence and uniqueness result we are aware.

A natural question is to see whether the Landau equation can be  directly derived, under a suitable scaling limit,  from
a particle system as it is the case of  the Boltzmann equation.
In fact one can see  (\cite {Bal}, see also \cite {Sp1} and \cite {P}), at a formal level, that the Landau equation is
expected to be valid for a weakly interacting  dense gas. The precise statement and scaling (called weak-coupling limit) will be
presented and discussed  in the next Section.
The formal analysis gives indeed the Landau equation \eqref{landau} with matrix \eqref{matrice_a}.
The two-body interaction potential $\phi$ is assumed smooth, spherically
symmetric,  and the constant $A$ is given by:
\begin{equation}
\label{A}
A=\frac 1{8\pi}\int_0^{+\infty} dr\, r^3 \hat \phi(r)^2,
\end{equation}
 where $\hat \phi(|k|) = \int\,dx\, \phi(|x|)e^{-ik\cdot x}$.

 Note that  we find the Landau equation with matrix \eqref{matrice_a}, which is not related to the Coulomb potential, but
 arises even though the potential is smooth and short-range. This fact was first established by N.N. Bogolyubov in 1946 \cite{N.N.B.}.

 In the present paper we want to start the rigorous analysis of the weak-coupling limit for an Hamiltonian particle system.
 Our result is very preliminary. We first decompose the $j$-particle marginals into two terms, one hopefully smooth
 and the other strongly oscillating, but small.  Eliminating this last term from the equations (with a procedure similar to that proposed by Zwanzig \cite{Zw})
 we find an equation with memory, which we can handled  up to the first order in time. We show that this contribution agrees with the corresponding one arising from the Landau equation. Roughly speaking we present a rigorous derivation of the Landau equation at time zero.

 It is well known that the situation for the Boltzmann equation is  better, namely we are able to derive such a kinetic equation for a short time \cite {Lanf}
(see also \cite{CIP} for additional comments and results)  in the low-density (or Boltmann-Grad) limit.

Note that the linear case, namely a single particle in a random potential under the weak-coupling limit, 
is well understood, see \cite {DGL} and references quoted therein.

Our analysis deals with the nonlinear problem but our techniques could apply as well to the linear case. We think that, while we can easily obtain  the same consistency result presented here, it seems very difficult to go further. In \cite {DGL} and related references, it is crucial the use of probabilistic tools which seems more efficient compared with the hierarchical approaches. In contrast  it is very difficult to implement the ideas working for the linear case to the present problem. 

Finally we want to mention that the same problem of characterizing the weak-coupling limit of particle systems, arises also in a quantum mechanical context.
In this case the quantity which we are interested in is the Wigner transform \cite {Wi} which is a way to describe a quantum state as a function in the classical phase space.
In contrast with the classical  case, we expect that the Wigner transform approaches, in the weak-coupling limit, the solution of a suitable Boltzmann equation, with a corrections due to the statistics, whenever taken in explicit consideration.
We quote \cite {HL}, \cite {BCEP1},  \cite {ESY},  \cite {BCEP2},  \cite {BCEP4} for the few results in this direction and  \cite {LS} and references quoted therein,  for the Boltzmann description of wave dynamics in the weak-coupling limit.

%%%%%%%%%%%%%%%%%%%%%%%%%%%%%%%%%%%%%%%%%%%%%%%
%%%%%%%%%%%%%%%%%%%%%%%%%%%%%%%%%%%%%%%%%%%%%%%

\section{Weak-coupling limit for classical systems}
\label{sec:1}

We consider a classical system of $N$ identical particles of unit
mass in the whole space. Positions and velocities are denoted by  the vectors $Q_N= \{q_1\dots q_N\} $ and
$V_N=\{v_1\dots v_N\}$ respectively.
The particles interact via a spherically symmetric, smooth  potential of finite range $\phi:\R^3 \to \R$, namely  $\phi(x)=0$ if $|x|>r$ for some positive $r$.  In the following we assume units for which $r=1$.

The Newton equations read as:
\begin{equation}
\label{MicroN}
\frac {d}{d\tau} q_i=v_i \qquad \frac {d}{d\tau} v_i=\sum _{\substack
{j=1 \dots N: \\ j\neq i}} F(q_i-q_j).
\end{equation}
Here $F=-\nabla \phi$ denotes the interparticle (conservative) force,
and
$\tau$ is the time.

Let $\var > r$  be a small parameter denoting the ratio between the macroscopic and microscopic space-time unities.

We are interested in a situation where the number of particles $N$ is very
large and the interaction strength quite moderate. The system has a unitary density so that
we assume $N=\var^{-3}$. In addition we look
for a reduced or macroscopic description of the system. Namely if $q$
and $\tau$ refer to the system seen in a microscopic scale, we rescale eq.n (\ref{MicroN})
in terms of the macroscopic variables
$$
x=\var q \qquad t=\var \tau
$$
whenever the physical variables of interest are varying on such scales
and are almost constant on the microscopic scales.

Remembering that we want to describe weakly interacting systems,
we also rescale
the potential according to:
\begin{equation}
\phi  \to \sqrt {\var}\phi ,
\end{equation}
so that
system (\ref{MicroN}), in terms of the $(x,t)$ variables,
becomes:
\begin{equation}
\label{MacroN}
\frac {d}{dt} x_i=v_i \qquad \frac {d}{dt} v_i=-\frac 1{\sqrt {\var}}\sum
_{\substack {j=1\dots N: \\ j\neq i}} \nabla \phi(\frac
{x_i-x_j}{\var})=\frac 1{\sqrt {\var}}\sum
_{\substack {j=1\dots N: \\ j\neq i}} F(\frac
{x_i-x_j}{\var}).
\end{equation}
Note that the velocities are automatically unscaled.

A statistical description of the above system passes through the introduction of
a probability distribution on the phase space
of the system.
Let $W^N=W^N (X_N,V_N)$ be a symmetric (in the exchange of variables) probability distribution. Here $(X_N,V_N)$ denote the set of positions and velocities:
$$
X_N=\{x_1 \dots x_N\} \quad V_N=\{v_1 \dots v_N\},  \quad x_i \in \R^3, v_i \in \R^3.
$$
Then from eq.ns (\ref {MacroN}) we obtain the following
Liouville equation
\begin{equation}
\label{Liouville}
(\pa_t+\sum_{i=1}^N v_i\cdot \nabla_{x_i}) W^N(X_N,V_N) =\frac 1{\sqrt\var} \big
(T^\var_NW^N\big)(X_N,V_N).
\end{equation}
Here we have introduced the operator
\begin{equation}
(T^\var_NW^N\big)(X_N,V_N)=\sum_{0<k<\ell\le N}
(T^\var_{k,\ell}W^N\big)(X_N,V_N),
\end{equation}
with
\begin{equation}
T^\var_{k,\ell}W^N=\nabla \phi(\frac {x_k-x_\ell}{\var})
\cdot (\nabla_{v_k}- \nabla_{v_\ell})W^N .
\end{equation}
To investigate the limit $\var \to 0$ it is convenient to introduce the
BBKGY hierarchy for the $j$- particle distributions defined as
\begin{align}
f^N_j(X_j,V_j)=&\int dx_{j+1} \dots \int dx_N \int dv_{j+1} \dots \int
dv_N \\
& W^N(X_j,x_{j+1} \dots x_N;V_j,v_{j+1} \dots v_N)\nn
\end{align}
for $j=1.\dots, N-1$.
Obviously we set $f^N_N=W^N$.
Note that BBGKY stands for Bogolyubov, Born, Green, Kirkwood and Yvon, the
names of physicists who introduced
independently this system of equations (see e.g. \cite {Bal}).

Such a hierarchy is obtained by means of
a partial integration of the Liouville equation (\ref{Liouville}) and standard
manipulations. The result is
(for $1\leq j\leq N$):
\begin{equation}
\label{hie}
(\pa_t+\sum_{k=1}^j v_k\cdot \nabla_{x_k})f^N_j=
\frac {1}{\sqrt {\var}} T^{\var}_j f^N_j + \frac {N-j}{\sqrt
{\var}}C^{\var}_{j+1}f^N_{j+1}
\end{equation}
for $1\leq j \leq N$.

We set
$$
f^N_j=0, \quad  \hbox {for}\quad j>N, \quad \hbox {and}\quad f^N_N=W^N.
$$
The operator $C^{\var}_{j+1}$ is defined as:
\begin{equation}
C^{\var}_{j+1}=\sum_{k=1}^j C^{\var}_{k,j+1}
\; ,
\end{equation}
and
\begin{align}
&C^{\var}_{k,j+1}f_{j+1}(x_1\dots x_j;v_1 \dots v_j)=\\
&-\int dx_{j+1}\int dv_{j+1}
F \left(\frac {x_k-x_{j+1}}{\var}\right)
\cdot \nabla_{v_k}
f_{j+1}(x_1,x_2,\dots,x_{j+1}; v_1,\dots, v_{j+1})
\nn .
\end{align}
 $C^\var_{k,j+1}$ describes the
interaction of particle $k$, belonging to the $j$-particle
subsystem, with a particle outside the subsystem, conventionally
denoted by the number $j+1$ (this numbering uses the fact that all
the particles are identical).

We finally fix the initial value  $\{f_j^0\}_{j=1}^N$ of the solution
$\{f^N_j(t)\}_{j=1}^N$  assuming that
$\{f_j^0\}_{j=1}^N$ is factorized, that is, for all $j=1,\dots N$
\begin{equation}
\label{chaos}
f^0_{j}=f_0^{\otimes j},
\end{equation}
where $f_0$ is a given one-particle distribution function.
This means that the state of any pair of particles is statistically
uncorrelated at time zero. Of course such a statistical independence is
destroyed at time $t>0$ because dynamics creates correlations and eq.n (\ref {hie})
shows that the time evolution of $f^N_1 $ is determined by the knowledge
of  $f^N_2 $ which turns out to be dependent on $f^N_3 $ and so on.
However, since the interaction between two given particles is going to
vanish in the limit $\var \to 0$, we can hope that such statistical
independence is recovered in the same limit.
Therefore we expect that when $\var \to 0$ the one-particle
distribution  function
$f^N_1$  converges to the solution of a suitable nonlinear kinetic
equation
$f$,  which we are going to investigate.

If we expand $f^N_j(t)$ as a perturbation of the free flow
$S(t)$ defined
as
\begin{equation}
(S(t) f_j)(X_j,V_j)= f_j(X_j-V_jt, V_j),
\end{equation}
we find
\begin{align}
\label{expand}
f_j^N(t)= &S(t)f_j^0+ \frac {N-j}{\sqrt {\var}}\int _0^t
S(t-t_1)C_{j+1}^\var f_{j+1}^N (t_1)dt_1+\\
&\frac {1}{\sqrt {\var}}\int _0^t S(t-t_1)T_{j}^\var f_{j}^N (t_1)dt_1\nn.
\end{align}
We now try to keep information on the limit behavior of $f_j^N(t)$.
Assuming for the moment that the time evolved $j$-particle
distributions $f_j^N(t)$ are smooth (in the sense that the first and second derivatives are
uniformly bounded in $\var$), then
\begin{align}
\label{coll1}
&C_{j+1}^\var f_{j+1}^N (X_j;V_j;t_1)=\\
&-\var^3 \sum_{k=1}^j
\int dr\int dv_{j+1}
F (r)
\cdot \nabla_{v_k}
f_{j+1}(X_j,x_{k}-\var r;V_j, v_{j+1},t_1)\nn.
\end{align}
Because of the identity
\begin{equation}
\int dr F(r)=0,
\end{equation}
we find that
\begin{equation}
C_{j+1}^\var f_{j+1}^N (X_j;V_j;t_1)=O(\var^4)
\end{equation}
provided that $D^2_v f_{j+1}^N$ is uniformly bounded. Since
$$
\frac {N-j}{\sqrt {\var}}=O(\var^{-\frac 72})
$$
we see that the second term in the right hand side of (\ref{expand}) does not give
any contribution in the limit.

Moreover
\begin{align}
\label{coll2}
&\int _0^t S(t-t_1)T_{j}^\var f_{j}^N (t_1)dt_1=\\
&\sum_{i \neq k}\int _0^t
dt_1 F \left(\frac {(x_i-x_k)-(v_i-v_k) (t-t_1)}{\var} \right)
\tilde f(X_j,V_j;t_1)\nn
\end{align}
where $\tilde f$ is a smooth function.
We note that the time integral in (\ref{coll2}) is $O(\var)$ because $F\neq 0$ only for times in an interval of length $O(\var)$. Therefore $f^N_j$ cannot be smooth since we expect a nontrivial limit.

In order to look for a (nontrivial)  kinetic equation, we can conjecture that
\begin{equation}
f_j^N=g_j^N+\gamma_j^N
\end{equation}
where $g_j^N$ is the main part of $ f_j^N$ and is smooth, while $\gamma_j^N$ is small, but strongly oscillating. We operate this decomposition according to the following equations which define $g_j^N$ and $\gamma_j^N$:
\begin{equation}
\label{newhie1}
(\pa_t+\sum_{k=1}^j v_k\cdot \nabla_{x_k})g^N_j=
 \frac {N-j}{\sqrt {\var}}C^{\var}_{j+1}g^N_{j+1}+
  \frac {N-j}{\sqrt {\var}}C^{\var}_{j+1}\gamma^N_{j+1}
\end{equation}
\begin{equation}
\label{newhie2}
(\pa_t+\sum_{k=1}^j v_k\cdot \nabla_{x_k})\gamma^N_j=
\frac {1}{\sqrt {\var}} T^{\var}_j \gamma^N_j + \frac {1}{\sqrt {\var}} T^{\var}_j g^N_j ,
\end{equation}
with initial data
\begin{equation}
g^N_j (X_j,V_j,0)=f^0_j (X_j,V_j), \quad \gamma^N_j (X_j,V_j)=0.
\end{equation}
Note that $\gamma^N_1=0$ since $T^{\var}_1=0$.

The remarkable fact of this decomposition is that $\gamma$ can be eliminated. Indeed, let
$$(X_j(t),V_j(t))=(\{x_1(t) \dots x_j(t), v_1(t) \dots v_j(t) \})$$ 
be the solution of the $j$-particle flow (in macro variables)
\begin{equation}
\label{jbody}
\frac {d}{dt} x_i=v_i \qquad \frac {d}{dt} v_i=-\frac 1{\sqrt {\var}}\sum
_{\substack {k=1\dots j: \\ k \neq i }} \nabla \phi\left(\frac
{x_i-x_k}{\var}\right),
\end{equation}
with initial datum $(X_j,V_j)=(\{x_1\dots x_j, v_1 \dots v_j \})$.
Denote by $U_j(t)$ the operator
\begin{equation}
\label{u}
U_j(t)f(X_j,V_j) =\exp \{ t( -\sum_i v_i \cdot \nabla_{x_i}  +\frac 1 {\sqrt{\e}} T_j) \} f(X_j,V_j)=  f(X_j(-t),V_j(-t)),
\end{equation}
then eq.n  (\ref {newhie2}) can be solved:
\begin{equation}
\label{gamma0}
\gamma^N_j (t)= \int_0^t ds U_j(s) \frac 1 {\sqrt {\var}} T_jg_j^N(t-s).
\end{equation}
Explicitly
\begin{equation}
\label{gamma}
\gamma^N_j (X_j,V_j,t)=-\frac 1 {\sqrt {\var}} \int_0^t ds
\sum
_{1\leq i <k \leq j} F \left(\frac
{x_i(-s) -x_k(-s) }{\var}\right) \cdot [(\nabla_{v_i}-\nabla_{v_k} )g_j^N](X_j(-s),V_j(-s);t-s).
\end{equation}
Inserting (\ref{gamma0})    in   (\ref{newhie1}) we finally arrive to a closed hierarchy for  $\{g_j^N\}_{j=1}^N$.
Obviously we pay the price of a memory term given by the time integral in (\ref{gamma0}) or  in (\ref{gamma}).

We write the hierarchy in integral form.
Then
\begin{align}
\label{Duha}
g_j^N(t)= &S(t)f_j^0+ \frac {N-j}{\sqrt {\var}}\int _0^t
S(t-\tau)C_{j+1}^\var  g^N_{j+1} (\tau) d\tau\\
&+\frac {N-j}{\sqrt{\var}}\int _0^t S(t-\tau) C^\var _{j+1}  \gamma_{j+1}^N (\tau ) d\tau \nn \\
=&S(t)f_j^0+ \frac {N-j}{\sqrt {\var}}\int _0^t
S(t-\tau)C_{j+1}^\var g_{j+1}^N (\tau) d\tau \nn \\
&+\frac {N-j}{ \var}\int _0^t  d\tau \int_0^\tau d\sigma S(t-\tau) C^\var _{j+1}  U_{j+1} (\tau-\sigma ) T_{j+1} g_{j+1}^N(\sigma). \nn
\end{align}
\medskip

\begin{remark}
Why do we expect that $\gamma^N_j$ strongly oscillates?
Let us try to control the first  derivatives of  $h(X_j,V_j,t)=U_j(t)
h_0(X_j,V_j)=h_0(X_j(-t),V_j(-t))$ for a given smooth function $h_0$. Then
$$
\frac {\pa h(X_j,V_j, t)}{\pa x_i^\alpha}=\sum_{ k,\beta}  \left( \frac {\pa
h_0}{\pa x_k^\beta} (X_j(-t),V_j(-t))
\frac {\pa x_k^\beta (-t)}{\pa x_i^\alpha}+  \frac {\pa h_0}{\pa v_k^\beta} (X_j(-t),V_j(-t))
\frac {\pa v_k^\beta (-t) }{\pa x_i^\alpha}\right)
$$
and analogous formula for  $\frac {\pa h(t)}{\pa v_i^\alpha}$. Here we are
using Greek indices for the components of $x_i$ and $v_i$.
To estimate  quantities like $\frac {\pa x_k^\beta (-t)}{\pa x_i^\alpha},
\frac {\pa x_k^\beta (-t)}{\pa v_i^\alpha}, \frac {\pa v_k^\beta (-t)}{\pa
x_i^\alpha}, \frac {\pa v_k^\beta (-t)}{\pa v_i^\alpha}$ we use eq.n
(\ref{jbody}) and find (changing $-t \to t$)
\begin{equation}
\label{djbody1}
\frac {d}{dt} \frac {\pa x_k^\beta (t)}{\pa x_i^\alpha}= \frac {\pa v_k^\beta (t)}{\pa
x_i^\alpha}, 
\end{equation} 
\begin{equation}
\label{djbody2}
 \frac {d}{dt} \frac {\pa v_k^\beta (t)}{\pa x_i^\alpha}
= \frac 1 {\var ^{3/2}}\sum
_{\substack {r=1\dots j: \\ r \neq k }}  \frac {\pa F^{\beta}}{\pa x_k^\gamma} (\frac
{x_k(t) -x_r (t) }{\var}) \left ( \frac {\pa x_k^\gamma (t)}{\pa x_i^\alpha} -
\frac {\pa x_r^\gamma (t)}{\pa x_i^\alpha} \right ).
\end{equation}

Integrating eq.ns \eqref {djbody1} and \eqref {djbody2} in time, we arrive,  by using the  Gronwall lemma,  to 
$$
\left|\frac {\pa v_k^\beta (t)}{\pa x_i^\alpha}\right| \leq C \exp \left(\frac {C \tau_c}  {\e^{3/2}}\right)
$$
where $\tau_c$ is the scattering time, namely the time interval for which $|x_k(t) -x_r (t)| \leq \e$ . Now, even though  $\tau_c= O(\e)$ (neglecting small relative velocities),  it seems difficult to get something better than a bound like
$ \exp (\frac C  {\sqrt {\e}})$.

In conclusion we expect that the first derivatives of $h(t)$ are
$ O(\exp (\frac 1 {\sqrt {\e}}))$.  Looking at  eq.n \eqref{gamma0} we expect for $\gamma$ the same behavior. In contrast,  the action of the operator $C_j$ is regularizing (althoug we are not able to prove this) so that  we expect $g$ to be smooth.

On the other hand $\gamma^N_j$ is also expected to be small, in some
sense. Indeed by taking the scalar product of (\ref {gamma0}) by a smooth
function $u$, we find
\begin{align}
|(u,\gamma^N_j (t))|& \leq \frac 1 {\sqrt {\var}} \int_0^t ds
\| U_j(-s) u \|_{L^\infty} \| T^\e_j g_j^N(t-s)\|_{L^1}  \nn \\
&\leq  \e^{5/2} \frac {j(j-1)}2 \|  u \|_{L^\infty}
 \int_0^t ds
\int dx_1 \int dx_3 \dots \int dV_j  \int dr |F(r)| \nn \\
&\ \ \ \ |(\nabla_{v_1}-\nabla_{v_2} )g_{j}^N
(x_1,x_1+\e r, x_3 \dots,V_j; t-s)| \nn.
\end{align}
Therefore this term is vanishing provided that $g^N$ is sufficiently
smooth (uniformly in $\e$). 
\end{remark}

A rigorous analysis of  the limit $N\to \infty$, $\var=N^{-(1/3)}$  seems to be very difficult. We expect that, in this limit, both  $f^N_j(t)$ and $g^N_j(t)$ would converge to $f(t)^{\otimes j}$, where $f$ solves the Landau equation stated in Introduction.
We cannot prove it, but a first step in this direction is made in the following Sections.

%%%%%%%%%%%%%%%%%%%%%%%%%%%%%%%%%%%%%%%%%%%%%%%
%%%%%%%%%%%%%%%%%%%%%%%%%%%%%%%%%%%%%%%%%%%%%%%

\section{Consistency}
\label{sec:2}

We consider eq.n \eqref{Duha} written in symbolic form as
\begin{equation}
g_j=S(t)f_j^0+A_{j+1}g_{j+1},
\end{equation}
where all upper indices $N$ are omitted for brevity. To solve these equations one can use the obvious iterative scheme
$$
g_j^0=S(t)f_j^0,\ \ \ g_j^{(n+1)}=S(t)f_j^0+A_{j+1}g_{j+1}^{(n)},\ \ \ n=0,1,\dots
$$
Our goal in this section is to prove that the equation for $g_1^{(1)}(t)=\tilde g_1^N(t)$ is consistent with the Landau equation.
Thus we replace (\ref {Duha})  by its first approximation:
\begin{align}
\label{Duhamel1}
\tilde g_j^N(t)= &S(t)f_j^0+ \frac {N-j}{\sqrt {\var}}\int _0^t
S(t-\tau)C_{j+1}^\var S(\tau) f^0_{j+1}  d\tau\\
&+\frac {N-j}{\sqrt{\var}}\int _0^t S(t-\tau) C^\var _{j+1}  \tilde \gamma_{j+1}^N (\tau ) d\tau \nn \\
=&S(t)f_j^0+ \frac {N-j}{\sqrt {\var}}\int _0^t
S(t-\tau)C_{j+1}^\var S(\tau)g_{j+1}^N d\tau\\
&+\frac {N-j}{ \var}\int _0^t  d\tau \int_0^\tau d\sigma S(t-\tau) C^\var _{j+1}  U_{j+1} (\tau-\sigma ) T_{j+1} S(\sigma)f^0_{j+1}. \nn
\end{align}
Here we set
\begin{align}
\label{gamma1}
\tilde \gamma^N_j (X_j,V_j,\tau )=&\frac 1 {\sqrt{\var}}
 \int_0^\tau d\sigma U_{j} (\tau-\sigma ) T_{j} S(\sigma)f^0_{j+1} \\
=& -\frac 1 {\sqrt {\var}} \int_0^t ds
\sum
_{1\leq i <k \leq j} F (\frac
{x_i(-s) -x_k(-s) }{\var}) \cdot  \\
&\cdot [(\nabla_{v_i}-\nabla_{v_k} ) S(\tau -s) f^0_j] (X_j(-s),V_j(-s)). \nn
\end{align}
We note that $\tilde \gamma^N_j$ can be explicitly computed.
\begin{lemma}
We have
\begin{align}
\label{expl}
\tilde \gamma^N_j(X_j,V_j,t) & = (U_j(t) f_j^0-S(t) f_j^0) (X_j,V_j) \nn \\
&=f^0_j( X_j (-t),V_j (-t)) - f^0_j( X_j-V_j t,V_j)\;.
\end{align}
\end{lemma}
\begin{proof}
Let ${\mathcal L}_0=-\sum_i v_i \cdot \nabla_{x_i} $  be the free flow generator.\\
Then we compute
\begin{align}
U_j(t)f^0_j- S(t) f^0_j=& \int_0^t ds
 \frac d {ds} [ U_j(s) S(t-s) ]f^0_j  \\
=&  \int_0^t ds
 [ U_j(s) ( {\mathcal L}_0 +\frac 1 {\sqrt{\e}} T_j)  S(t-s) ]f^0_j   \nn \\
 &- \int_0^t ds
 [ U_j(s)  {\mathcal L}_0  S(t-s) ]f^0_j \nn \\
 =& \tilde \gamma_j^N (t).   \quad \quad \quad \quad \quad \quad \quad \quad \quad \quad\text      \nn 
\end{align} 
\end{proof}
For convenience of the reader we make explicit eq.n (\ref {Duhamel1})  in the case $j=1$
\begin{align}
\label{Duhamel2}
\tilde g_1^N(t)= &S(t)f_0+ \frac {N-1}{\sqrt {\var}}\int _0^t
S(t-\tau)C_{2}^\var S(\tau)f_{2}^0 d\tau\\
&+\frac {N-1}{\sqrt{\var}}\int _0^t S(t-\tau) C^\var _{2}  \tilde \gamma^N_{2} (\tau ) d\tau \nn
\end{align}
where, by Lemma 1,
\begin{align}
\label{gamma3}
\tilde \gamma^N_2 (x_1,v_1,x_2, v_2  ,\tau )=& -\frac 1 {\sqrt {\var}} \int_0^\tau ds
F (\frac
{x_1(-s) -x_2(-s) }{\var}) \cdot  \\
& \cdot[(\nabla_{v_1}-\nabla_{v_2} ) S(\tau -s) f^0_2](X_2(-s),V_2(-s))  \nn \\
=& \big[ f^0_2( X_2 (-\tau),V_2 (-\tau))-f_2^0(X_2-V_2\tau,V_2) \big]. \nn
\end{align}
The first result of the present paper is summarized in the following Theorem.
\begin{theorem}
Suppose $f_0\in C^3_0(\R^3 \times \R^3)$ be the initial probability density
satisfying:
\begin{equation}\label{phi-property}
| D^r f_0 (x,v)| \leq C e^{-b |v|^2} \quad \text {for} \qquad  r=0,1, 2
\end{equation}
where $D^r$ is any derivative of order $r$ and $b >0 $. 
Assuming also that $\phi \in C^2(\R^3)$  and $\phi(x)=0$ if $|x| >1$.
If \eqref{chaos} holds for $j=1,2$,
 then
\begin{equation}
\label{result}
\lim_{\var \to 0} \tilde g^N_1(t)=S(t) f_0+\int_0^t d\tau S(t-\tau) Q_L( S(\tau)f_0, S(\tau)f_0),
\end{equation}
\begin{equation}
\label{resultgamma}
 \lim_{\var \to 0} \tilde \gamma^N_1(t)=0,
\end{equation}
where $N\var^3=1$ and the above limits are considered in ${\mathcal D}'$.
\end{theorem}
\begin{proof}
Let $u \in {\mathcal D} (\R^3 \times \R^3)$  be a test function.
From now on we will denote by $(h_j,k_j)=\int dX_j\int dV_j\, h_j(X_j,V_j)k_j(X_j,V_j)$ the inner product.
Then
\begin{align}
\label{wDuhamel}
(u,\tilde g_1^N(t))= (u,S(t)f_0)+ \frac {N-1}{\sqrt {\var}}\int _0^t
(u, S(t-\tau)C_{2}^\var S(\tau)f_{2}^0) d\tau_1+
\int _0^t  {\mathcal T}_\var (\tau)d\tau,
\end{align}
where
\begin{align}
\label{tau1}
{\mathcal T}_\var (\tau)=& -\frac {N-1}{\sqrt{\var}}  \int dx_1 \int dx_2 \int dv_1\int dv_{2}
(\nabla_{v_1} S(\tau-t) u(x_1,v_1)) \cdot \\
 &\cdot F \left(\frac {x_1-x_{2}}{\var}\right)\,
\tilde \gamma_2^N  (x_1,x_2,v_1,v_2,\tau) \nn.
\end{align}\\
We have already seen that the second term in the right hand side of  (\ref{wDuhamel}) is vanishing. Therefore we have to evaluate the last term, namely
$
\int_0^t  d\tau {\mathcal T}_\var (\tau).
$
We split the term ${\mathcal T}_\var (\tau)$ into two terms
\begin{equation}
\label{decomp}
{\mathcal T}_\var= {\mathcal T}_\var ^{\leq}+{\mathcal T}_\var^{>}
\end{equation}
where
\begin{align}
\label{tau>}
{\mathcal T}_\var (\tau)^>=&- \frac {N-1}{\sqrt{\var} } \int dx_1 \int dx_2 \int  dv_1\int_{|w|>a \e^{1/4}} dv_{2}
(\nabla_{v_1} S(\tau-t) u(x_1,v_1))\cdot \\
 &\cdot F \left(\frac {x_1-x_{2}}{\var}\right)
\tilde \gamma_2^N  (x_1,x_2,v_1,v_2,\tau) \nn
\end{align}
where $w=v_1-v_2$ is the relative velocity and $a$ is a number to be fixed later on.
${\mathcal T}_\var ^{\leq}$ is defined accordingly.\\
The reason of this decomposition  will be clear  later on. For the moment we show that
${\mathcal T}_\var ^{\leq}$  is negligible.

\begin{lemma}
\begin{equation}
 {\mathcal T}_\var ^{\leq} =O(\e^{1/4})\;.
\end{equation}
\end{lemma}

\begin{proof}
By Lemma 1 we have that   $\tilde \gamma^N_2$ is uniformly bounded. Moreover by the change of variables
$$
x_2=x_1-\e r
$$
we get
\begin{align}
| {\mathcal T}_\var ^{\leq}| \leq & C (N-1) \e^3 \frac 1 {\sqrt {\e}} \int dx_1  \int  dv_1|\nabla_{v_1} S(\tau-t) u(x_1,v_1)|
 \int dr |F(r)|  \int_{|w| \leq a \e^{1/4}} dw \\
\leq  & C \e^{1/4}. \nn   
 \end{align}
 \end{proof}
To evaluate ${\mathcal T}_\var^{>}$ we use (\ref{gamma3}) to write it as
\begin{align}
\label{tau0}
{\mathcal T}_\var^{> }(\tau)= &(N-1)\var^3  \int dx_1 \int dr \int dv_1\int_{|w|>a \var^{1/4}} dv_{2} \\
& \frac 1 {\var}  \int_0^{\tau} ds F_\alpha( r) F_\beta\left( \frac  {x_1( -s)-x_2(-s)} {\var}\right)
 [h_\var (x_1,x_2,v_1,v_2,\tau, s)]_{\alpha,\beta},  \nn
 \end{align}
 where $x_2=x_1-\e r$ and $h_\var$ is the matrix
$$
(h_\var)_{\alpha,\beta}= -(\nabla_{v_1} S(\tau-t) u(x_1,v_1))_\alpha [  (\nabla_{v_1}-\nabla_{v_2} )
S(\tau-s) f^0_2]_\beta (X_2(-s),V_2(-s)), \ \ \alpha,\beta=1,2,3.
$$
The summation over repeated Greek indices is assumed here and below.\\
Here the flow $X_2 (t)=( x_1(t),x_2(t) )$ has initial conditions $(x_1, x_1-\e r)$.
Scaling times we also find
\begin{align}
\label{tau1}
{\mathcal T}_\var^{> }(\tau)= &(N-1)\var^3  \int dx_1 \int dr \int dv_1\int_{|w|>a \var^{1/4}} dv_{2} \\
& \int_0^{\tau/\e } ds F_{\alpha}( r) F_{\beta}\left( \frac  {x_1( -\e s)-x_2(-\e s)} {\var}\right)
 [h_\var (x_1,x_2,v_1,v_2,\tau, \e s)]_{\alpha,\beta}\,.  \nn
 \end{align}
Let us   introduce the function $h$ which is the formal limit of $h_\var$, namely
\begin{equation}\label{h}
h_{\alpha,\beta}(x_1,v_1,v_2,\tau) =- R_{\alpha}(x_1,v_1,\tau) [  (\nabla_{v_1}-\nabla_{v_2} )
S(\tau) f^0_2(x_1,x_1,v_1,v_2)]_{\beta},
\end{equation}
where
\begin{equation}\label{test-f}
R(x_1,v_1,\tau)= \nabla_{v_1} S(\tau-t) u(x_1,v_1).
\end{equation}
We split  ${\mathcal T}_\var^{> } $ into two terms
$$
{\mathcal T}_\var^{> }={\mathcal T}_1^{> }+{\mathcal T}_2^{> }
$$
where
\begin{equation}
\label{t1}
\begin{split}
{\mathcal T}_1^{> } (\tau)=(N-1)\var^3  &\int dx_1 \int dr \int dv_1\int_{|w|>a \var^{1/4}} dv_{2}\\
 &\int_0^{\frac {\tau}{\var}} ds F_{\alpha}( r) F_{\beta}\left( \frac  {x_1(-\var s)-x_2(-\var s)} {\var}\right)
 h_{\alpha,\beta} (x_1,v_1,v_2,\tau)
\end{split}
\end{equation}
and
\begin{equation}
\label{t2}
\begin{split}
{\mathcal T}_2 ^{>} (\tau)=(N-1)\var^3  &\int dx_1 \int dr \int dv_1\int_{|w|>a \var^{1/4}} dv_{2}\\
 &\int_0^{\frac {\tau}{\var}} ds F_\alpha( r) F_\beta\left( \frac  {x_1(-\var s)-x_2(-\var s)} {\var}\right)
(h_\var - h)_{\alpha,\beta} .
\end{split}
\end{equation}
We shall show that ${\mathcal T}_2 ^{>} (\tau)$ is vanishing while ${\mathcal T}_1 ^{>} (\tau)$ has the right behavior.
In the evaluation of  ${\mathcal T}_1 ^{>} (\tau)$ we note that $h$ does not depend on $s$ so that we have to evaluate the integral
\begin{equation}
\label{int}
 \int_0^{\frac {\tau}{\var}} ds  F\left( \frac  {x_1(-\var s)-x_2(-\var s)} {\var}\right)=\frac 1 \e  \int_0^{\tau} ds F\left( \frac  {x_1( -s)-x_2(-s)} {\var}\right).
\end{equation}
Indeed the integral (\ref {int}) can be bounded when the interaction time of the two-particle system is $O(\e)$ and this is  true only
if  the relative velocity is not too small (see Lemma 3 below). This explains why we did the decomposition \eqref {decomp}.

\begin{lemma}
Setting  $w= v_1-v_2$,  suppose that
\begin{equation}\label{high_relv}
|w|> a \var^{1/4}
\end{equation}
where $a=4\sqrt{ \| F\|_{L^\infty}}$. Then, defining for any real number $s$
\begin{equation}
\Delta_\var =   \{ s | |x_1(s)-x_2(s)|< \var  \} ,
\end{equation}
we have
\begin{equation}
\emph{meas} (\Delta_\var)   \leq \frac {4\var}{|w|} .
\end{equation}
Moreover, for $i=1,2$:
\begin{equation}
\label{more}
|v_i(\var s) -v_i| \leq C \frac {\sqrt {\var}} {|w|}.
\end{equation}
\end{lemma}
\begin{proof}
Assuming first that $s>0$, we pass in the coordinate system around the center of mass (at the origin) and denote by
$\xi(t)=x_1(t)-x_2(t)$. Let $w=v_1-v_2$ be the relative velocity and $w_x$ its horizontal component.
We assume that at time zero the particles are in the interaction disk (more precisely, they enter in the interaction disk at time $s=0$) and fix the axis
in such a way that $w$ is  horizontal and its $x$- component is positive, namely $w_x=|w|$. Let $\bar t$ be the first  time for which
$$
w_x(t) \leq \frac {|w|}2.
$$
By the equation of motion
\begin{equation}
w_x(t) =|w| + \int_0^t  \frac 2 {\sqrt {\var}} F_x \left(\frac {\xi(s)}{\var}\right)
\end{equation}
we infer
$$
\frac {|w|}2   \geq |w| - \frac 2 {\sqrt {\var}} \| F \|_{L^\infty} \bar t
  $$
from which
 \begin{equation}
\bar t \geq \frac {\sqrt {\var}|w|} {4 \| F \|_{L^\infty} }.
\end{equation}
In the time interval $[0,\bar t]$ we have $w_x \geq \frac {|w|}2$ and the horizontal displacement is (under assumption \eqref{high_relv}) larger than
\begin{equation}\label{t-bar}
\frac {|w|}2\bar t \geq 2\var,
\end{equation}
since the diameter $2\e$ is a maximal path inside the sphere, independent of the initial point.\\
This implies  that, when $|\xi(t)| <\var$,  then $|w(t)| > |w(0)|/2$ and hence
\begin{equation}
\mbox{meas} (\Delta_\var)  \leq \frac {4\var} { |w|}.
\end{equation}
Moreover
\begin{equation}
v_1(\var s ) =v_1 + \int_0^{\var s}  \frac 1 {\sqrt {\var}} F\left(\frac {x_1(\sigma )-x_2(\sigma )}{\var}\right) d\sigma
\end{equation}
from which
\begin{equation}
|v_1(\var s ) -v_1| \leq C \frac {\sqrt {\var}} {|w|}.
\end{equation}
The case $s<0$ reduces to the case $s>0$ by changing the initial velocities to $v_i(0)=-v_i$ for $i=1,2$. This completes the proof. 
\end{proof}

Note  that
\begin{equation}
\frac {x_1(-\var s)-x_2(-\var s)} {\var}=r-ws +\frac 1 \var \int_0^{-\var s} d\sigma
[(v_1 (\sigma)-v_1) - (v_2 (\sigma)-v_2)]
\end{equation}
thus, by Lemma 3,
\begin{equation}
\label{lip}
\left|\frac {x_1(-\var s)-x_2(-\var s)} {\var}-(r-ws)\right| \leq Cs  \frac {\sqrt {\var}} { |w|}.
\end{equation}
The integral \eqref{t1} reads
\begin{align}
{\mathcal T}_1 ^{>} (\tau) & =(N-1)\var^3  \int dx_1 \int dr \int dv_1\int_{|w|>a \var^{1/4}} dv_{2} \\
&
 \int_0^{\frac {\tau}{\var}} ds F_{\alpha}( r) F_{\beta}( r-ws)
 h_{\alpha,\beta} (x_1,v_1,v_2,\tau)+E \nn
\end{align}
where the error term $E$ is given by
\begin{align}
E & =(N-1)\var^3  \int dx_1 \int dr \int dv_1\int_{|w|>a \var^{1/4}} dv_{2} \\
&
 \int_0^{\frac {\tau}{\var}} ds F_{\alpha}( r)  \left[   F_\beta\left( \frac  {x_1(-\var s)-x_2(-\var s)} {\var}\right)  -F_\beta ( r-ws)  \right]
 h_{\alpha,\beta} (x_1,v_1,v_2,\tau)  \nn
\end{align}
It is clear from the proof of Lemma 3 that $|x_1(-\e s)-x_2(-\e s)|\geq\e$ if $s\geq 4/|w|$ (see \eqref{t-bar}). On the other hand, $|r-ws|\geq 1$ if $s\geq 2/|w|$, provided $|r|\leq 1$. Hence,
\begin{align}
|E|\leq & C \sqrt {\e} \int dx_1 \int dv_1 \int_{|w|>a \var^{1/4}} dv_{2} \frac 1 {|w|}  \int_0^{\frac {4}{|w|}} s ds
|h (x_1,v_1,v_2,\tau)|\\
\leq &  C \sqrt {\e}  \int dx_1 \int dv_1 \int_{|w|>a \var^{1/4}} dv_{2} \frac 1 {|w|^3} |h (x_1,v_1,v_2,\tau)| \nn \\
\leq & C \sqrt {\e} | \log \e| . \nn
 \end{align}
 In the last step we estimated
 \begin{align}
 |h(x_1,v_1,v_2,\tau)\leq & C |f_0 (x_1 -v_2 \tau , v_2) ( \nabla_{v_1}-\tau \nabla_{x_1} )f_0 (x_1-v_1 \tau, v_1) | \\
 &+
| f_0 (x_1 -v_1 \tau, v_1) ( \nabla_{v_2}-\tau \nabla_{x_2}) f_0 (x_1-v_2 \tau, v_2) | \nn  \\
\leq & C e^{- b (|v_1|^2+|v_2|^2)} \nn.
 \end{align}
 
\begin{lemma}
For all $w \neq 0$,
\begin{equation}
\label{lim11}
\lim_{\var \to 0} \int dr  \int_0^{+\tau/\var} ds F_\alpha( r) F_\beta( r-ws)=
\frac 12 \lim_{\var \to 0} \int dr  \int_{-\tau/\var}^{+\tau/\var} ds F_\alpha( r) F_\beta( r-ws)= a(w)_{\alpha,\beta}
\end{equation}
where
\begin{equation}
a(w)_{\alpha,\beta}=\frac A{|w|}(\delta_{\alpha,\beta} -\frac {w_\alpha w_\beta}{|w|^2})
\end{equation}
and
\begin{equation}
A=\frac{1}{8\pi}\int_0^\infty d\rho \rho^3 \hat \phi^2 (\rho),
\end{equation}
with $\hat\phi(|k|)=\int_{\R^3}\phi(r)e^{-ik\. r}$.
\end{lemma}

\begin{proof}
The first identity in (\ref {lim11}) is due to the symmetry $F(r)=F(-r)$. Then
we compute the left hand side of (\ref {lim11}) taking the Fourier transform and passing in spherical coordinates. The result is
\begin{equation}
A\int_{S^2} d\hat k \delta (\hat k \cdot w) \hat k \otimes \hat k=a(w).  
\end{equation}
\end{proof}

Finally by the use of the dominated convergence theorem we can establish
\begin{equation}
\lim_{\var \to 0}\int d\tau {\mathcal T}_1^>(\tau)=\int_0^t d\tau S(t-\tau) Q_L( S(\tau)f_0, S(\tau)f_0)
\end{equation}
in ${\mathcal D}'$.
To conclude the proof it remains to show that
\begin{equation}\label{t2_0}
\lim_{\var \to 0}\int _0^t d\tau { \mathcal T}_2^>(\tau)=0.
\end{equation}
We first evaluate
\begin{equation}
\begin{split}
(h_\var-h&)_{\alpha,\beta} (x_1,r, v_1,v_2,\tau, \var s)=R_\alpha(x_1,v_1,\tau )\\
&\{[ (\nabla_{v_1}-\nabla_{v_2} )S(\tau-\var s)
f^0 _{2} ](X_2(-\var s),V_2(-\var s))-(\nabla_{v_1}-\nabla_{v_2} )S(\tau) f^0_2(x_1,x_1,v_1,v_2)]\}_\beta\,.
\end{split}
\end{equation}
Note that
$$
\nabla_v S(\tau) f(x,v)=S(\tau)(\nabla_v-\tau \nabla_x)f(x,v).
$$
Omitting irrelevant variables we observe that
$$
(h_\e-h)_{\alpha,\beta}=R_\alpha(\Phi_\beta(-\e s)-\Phi_\beta(0))
$$
where $\Phi(\sigma)=[(\nabla_{v_1}-\nabla_{v_2})S(\tau+\sigma)f_2^0](X_2(\sigma),V_2(\sigma))$.\\
Hence
$$
|h_\e-h|\leq |R|\int_{-\e s}^0 d\sigma |\dot{\Phi}(\sigma)|.
$$
It is easy to see that $\dot{\Phi}(\sigma)$ is a linear combination of various second derivatives of $f_2^0$, multiplied by $\dot{w}(\sigma)=\frac{2}{\sqrt \e}F\left(\frac{x_1(\sigma)-x_2(\sigma)}{\e}\right)$,
plus two terms proportional to first derivatives with respect to $x$. All the derivatives are computed at the point $[X_2(\sigma)-(\tau+\sigma)V_2(\sigma),V_2(\sigma)]$. Hence, under the assumptions of Theorem 1, we obtain
$$
|h_\e-h|\leq C|R|\frac{1}{\sqrt \e}\int_0^{\e s}d\sigma \exp\{-b (|v_1(-\sigma)|^2+ |v_2(-\sigma)|^2)\}.
$$
Since $|x_1(-\e s)-x_2(-\e s)|\geq \e$ if $s\geq 4/|w|$, the integral over $ds$ in \eqref{t2} can be estimated by
\begin{equation}\label{F}
\begin{split}
\left| \int_0 ^{\tau/\e} ds  F_\beta\left(\frac{x_1(-\e s)-x_2(-\e s)}{\e}\right)(h_\e-h)_{\alpha,\beta}\right|&\leq \frac{C|R|}{\sqrt\e}\int_0^{4/|w|}ds\int_0 ^{\e s}d\sigma \psi(\sigma)\\
& \leq \frac{C|R|\sqrt\e}{|w|}\int_0 ^{4/|w|} d\sigma \psi(\e\sigma)\;,
\end{split}
\end{equation}
where $\psi(\sigma)=\exp\{-b(|v_1(-\sigma)|^2+ |v_2(-\sigma)|^2)\}.$ 

Then by energy conservation
$$
|v_1(t)|^2+|v_2(t)|^2+2\sqrt\e\phi\left(\frac{x_1(t)-x_2(t)}{\e}\right)=const
$$
and therefore
$$
\psi(\e s)\leq A^2 \exp\{-b(|v_1|^2+|v_2|^2-4\sqrt\e \|\phi\|_\infty)\}.
$$
Hence, we obtain the following estimate of the integral \eqref{t2}:
\begin{equation*}
\begin{split}
|\mathcal T_2^>(\tau)| \leq C\sqrt\e \int dx_1 \int dr \int dV_2\, |v_1-v_2|^{-2}|R(x_1,v_1,\tau)|\,|F(r)|\, \exp\{-b(|v_1|^2+|v_2|^2)\},
\end{split}
\end{equation*}
where $R(x_1,v_1,\tau)$ is given in \eqref{test-f}. It is clear that $R(x_1,v_1,\tau)=0$ if $|x_1|>R_1$, where $R_1$ depends only on $\tau$. Therefore
$$
|\mathcal T_2^>(\tau)|\leq C\sqrt\e \int dv_1\int dv_2 |v_1-v_2|^{-2}\exp\{-b(|v_1|^2+|v_2|^2)\}=C_1\sqrt\e.
$$
 By Lemma 1 we also conclude that  \eqref{resultgamma}   holds and this completes the proof of Theorem 1.  
 \end{proof}

%%%%%%%%%%%%%%%%%%%%%%%%%%%%%%%%%%%%%%%%%%%%%%%
%%%%%%%%%%%%%%%%%%%%%%%%%%%%%%%%%%%%%%%%%%%%%%%

\section{Propagation of chaos}
\label{sec:3}

In this section we extend the result obtained in Theorem 1 to the $j$-marginal distribution, showing the propagation of chaos (at first order in time). More precisely we have

\begin{theorem}
Under  hypotheses of Theorem 1, if \eqref{chaos} holds for all j, then
\begin{equation}\label{thm2}
\begin{split}
\lim_{\e\rightarrow 0}\tilde{g}_j ^N (t, x_1,v_1,\dots,x_j,v_j)&=\prod_{i=1}^j S(t)f_0(x_i,v_i)+\\
&+\sum_{i=1}^j \prod_{\substack{k=1\\k\neq i}}^jS(t)f_0(x_k,v_k)\int_0^t d\tau S(t-\tau)Q_L(S(\tau)f_0,S(\tau)f_0)(x_i,v_i),
\end{split}
\end{equation}
\begin{equation}
\label{thm2gamma}
\lim_{\e\rightarrow 0}\tilde{\gamma}_j ^N (t, x_1,v_1,\dots,x_j,v_j)=0
\end{equation}
in $\mathcal{D'}$.
\end{theorem}

\begin{remark}
The reason why we call eq.n \eqref{thm2} propagation of chaos is that the r.h.s of \eqref{thm2} corresponds to the first order in time of $\prod_{i=1}^j\lim _{\e \to 0} \tilde g_1^N(x_i,v_i,t)$.
\end{remark}

\begin{proof}
Let  $u\in \mathcal D(\R^{3j}\times\R^{3j})$ be a test function and let us consider
\begin{equation}\label{g2}
\begin{split}
(u,\tilde g_j ^N(t))=(u,S(t)f_j ^0)&+\frac{N-j}{\sqrt\e}\int_0 ^t (u,S(t-\tau)C^\e _{j+1}
S(\tau)f_{j+1} ^0) d\tau \\
 &+\frac{N-j}{\sqrt\e}\int_0 ^t (u,S(t-\tau)C^\e _{j+1} \tilde\gamma _{j+1}^N (\tau))d\tau.
\end{split}
\end{equation}
Of course  the second term in  \eqref{g2} is of order $O(\sqrt\e)$, hence we focus on the third term. Then, defining $R_i(X_j,V_j,\tau):=\nabla_{v_i}S(\tau-t) u(X_j,V_j)$, such a term is 
\begin{equation}
\begin{split}
\label{T}
T=& \frac{N-j}{\sqrt\e}\int_0 ^t (u,S(t-\tau)C^\e _ {j+1}(\tau)\tilde \gamma_{j+1} ^N) d\tau \\ 
&=-\frac{N-j}{\sqrt\e}\sum_{i=1} ^j \int_0 ^t d\tau \int dX_{j+1}\int dV_{j+1}\, R_i(X_j,V_j,\tau)\. F(\frac{x_i-x_{j+1}}{\e})\tilde\gamma _{j+1} ^N(\tau) \\ 
&=-\frac{N-j}{\e}\sum_{i=1} ^j \int_0 ^t d\tau \int dX_{j+1}\int dV_{j+1} \int_0 ^\tau ds R_i(X_j,V_j,\tau) \sum_{\substack{k,l=1\\ k< l}}^{j+1}F\left(\frac{x_i-x_{j+1}}{\e}\right)\\
&\ \ \ \ \{F\left(\frac{x_k(-s)-x_l(-s)}{\e}\right)\.[(\nabla_{v_k}-\nabla_{v_l})S(\tau-s)f_{j+1} ^0](X_{j+1}(-s),V_{j+1}(-s) )\}. \\
\end{split}
\end{equation}

We shall see that the leading term in the sum appearing in the r.h.s. of \eqref {T} is that with $k=i, l=j+1 $, the other ones being vanishing. This is the content of the following

\begin{lemma}
Let $\varphi=\varphi (X_{j+1},V_{j+1}, \tau, s)\geq 0$ be a measurable function, compactly supported in $X_{j+1}$ and such that
$$
\varphi \leq e^{-b |V_{j+1}|^2}.
$$
Then, if $(k,l) \neq (i,j+1)$, for all $i,k,l$, we have
\begin{equation}
\label{T1}
\frac{N-j}{\e}\int _0^{\tau} ds \int dX_{j+1}\int dV_{j+1} \varphi 
|F\left(\frac{x_i-x_{j+1}}{\e}\right)| |F\left(\frac{x_k(-s)-x_l(-s)}{\e}\right) | \leq
C^j \e\,.
\end{equation}
\end{lemma}
\begin{proof}
We are integrating on the final coordinates $(X_{j+1}, V_{j+1})=(X_{j+1}(0), V_{j+1}(0) ) $ of the flow  $(X_{j+1}(\sigma), V_{j+1} (\sigma) )$ defined for negative times 
$\sigma \in [-\tau,0]$. We find convenient to reverse the velocities $V_{j+1} \to -V_{j+1}$ and look at positive times  $s \in [0, \tau]$.

First of all we perform the usual change of variables $x_{j+1}=x_i-\e r$ and gain $\e^3$. Next we introduce the following partition of the phase space:
setting $C_0= \{ (k,l), k<l | (k,l)\neq (i,j+1)\} $ we define
\begin{equation}
A_0 (k,l)=\{ (X_{j+1},V_{j+1}) | \quad |x_k-x_l|< 2 \e , (k,l)\in C_0 \}
\end{equation}
and
\begin{equation}
A_0= \bigcup_{(k,l) \in C_0}  A_0(k,l).
\end{equation}

Furthermore, denoting by $s(k,l) \in [0,\tau ] $ the first instant for which 
\begin{equation}
 |x_k(s) -x_l(s) |<  \e 
\end{equation}
namely the pair of particles $k$ and $l$ starts to interact at time $s(k,l)$ (if they do not interact we set $s(k,l)=\tau$) 
we define:
\begin{equation}
A_{k,l}=\{ (X_{j+1},V_{j+1}) \notin A_0  | s(k,l) =\min_ {(r,m) \in C_0} s(r,m)< \tau \}.
\end{equation}
In other words if $ (X_{j+1},V_{j+1}) \in A_{k,l}$ the pair of particles $(k,l)\in C_0$ is the first interacting pair 
(excluded the pair $(i,j+1)$ which starts to interact at time $0$) in the time interval $(0,\tau]$.

Note that we are interested to integrate over the set
\begin{equation}
\label{set}
A_0\cup  \bigcup_{(k,l) \in C_0}  A_{k,l}.
\end{equation}
In facts in the complement of the set  \eqref{set}, \eqref{T1} vanishes because
$$
|F\left(\frac{x_k(s)-x_l(s)}{\e}\right) | =0.
$$

To estimate the contribution due to $A_{k,l}$ we first assume that $k\neq i, l\neq j+1,i $.

Note that the motion of the pair of particles with indices $(k,l)$ is free in $[0,s(k,l)]$. Then setting
$x_k-x_l=y$ and $v_k-v_l=w$ we have
\begin{equation}
\label{con}
\inf _{s\in [0,\tau] } |y-ws| \leq \e.
\end{equation}
The minimizing $s$ is $s_0=\frac {w\cdot y}{|w|^2}$ so that  condition \eqref{con} yields
\begin{equation}
|y-w \frac {w\cdot y}{|w|^2}| \leq \e.
\end{equation}
This means that the projection of $y$ on the orthogonal plane to $w$ is in the disk smaller than $\e$. Therefore
\begin{equation}
\label{T2}
\frac{N-j}{\e} \int_{A(k,l)}  dX_{j+1} dV_{j+1} \, \varphi \,
|F\left(\frac{x_i-x_{j+1}}{\e}\right)| |F\left(\frac{x_k(-s)-x_l(-s)}{\e}\right) | \leq
C^j \e\,.
\end{equation}

Now we consider the cases  $k=i$, $l=i$ or $l=j+1$. For the sake of clearness we consider  $k=i$,  the other cases  being
completely analogous.

There are two possibilities: either $s(i,l)> \tilde s$, where $\tilde s$ is the last interaction time for the pair
$(i,j+1)$, namely
$$
|x_i(s)-x_{j+1}(s) | >\e
$$
for $s>\tilde s$, or  $s(i,l)\leq  \tilde s$.\\
In the first  case we can repeat the above argument  setting $y=x_i(\tilde s)-x_l(\tilde s)$ and
$w=v_i(\tilde s)-v_l(\tilde s)$. \\
In the second one observe that the center of mass $\bar x=\frac {x_i+x_{j+1}}{2}$ is moving freely with velocity 
$\bar v=\frac {v_i+v_{j+1}}{2}$ (because the pair $(i,j+1)$ is an isolated system at least up to a time $\bar s=s(i,l)$).

Condition
$$
|x_i(\bar s)-x_{l}(\bar s) | = \e
$$
implies
\begin{equation}
\label{cm}
|x_l(\bar s)-\bar x(\bar s) | \leq |x_i(\bar s)-x_{l}(\bar s) | + |x_i(\bar s)-\bar x(\bar s) | \leq \frac {3} {2}  \e\;.
\end{equation}
Therefore we can integrate under the condition \eqref{cm} to get 
\begin{equation}
\label{T3}
\frac{N-j}{\e} \int_{A(i,l)}  dX_{j+1} dV_{j+1} \, \varphi \,
|F\left(\frac{x_i-x_{j+1}}{\e}\right)| |F\left(\frac{x_k(-s)-x_l(-s)}{\e}\right) | \leq
C^j \e\,.
\end{equation}

Clearly we also have that
\begin{equation}
\label{T4}
\frac{N-j}{\e} \int_{A_0 }  dX_{j+1} dV_{j+1} \, \varphi \,
|F\left(\frac{x_i-x_{j+1}}{\e}\right)| |F\left(\frac{x_k(-s)-x_l(-s)}{\e}\right) | \leq
C^j \e^2\,.
\end{equation}
Thus we conclude the proof.  
\end{proof}
Finally we handle the leading term. Setting
\begin{equation}
\begin{split}
\label{Tl}
T_l=& -\frac{N-j}{\e}\sum_{i=1} ^j \int_0 ^t d\tau \int dX_{j+1}\int dV_{j+1} \int_0 ^\tau ds R_i(X_j,V_j,\tau)F\left(\frac{x_i-x_{j+1}}{\e}\right)\\
&\left\{F\left(\frac{x_i(-s)-x_{j+1}(-s)}{\e}\right)\.
[(\nabla_{v_k}-\nabla_{v_l})S(\tau-s)f_{j+1} ^0](X_{j+1}(-s) ,V_{j+1}(-s) )\right\},
\end{split}
\end{equation}
we have

\begin{lemma}
The term with repeated indices is of order one.
More precisely,
\begin{equation}
\label{result2}
\begin{split}
\lim_{\e\rightarrow 0} T_l =\left ( \sum_{i=1}^j \prod_{\substack{k=1\\k\neq i}}^jS(t)f_0(x_k,v_k)\int_0^t d\tau S(t-\tau)Q_L(S(\tau)f_0,S(\tau)f_0)(x_i,v_i),u\right).\end{split}
\end{equation}
\end{lemma}

\begin{proof}
At this point the proof is rather obvious and we only sketch it. We first reduce the integration domain in the definition of $T_l$ for moderately  large relative velocity, i.e. $|v_i-v_{j+1}|> a\e^{1/4} $,
being the contribution of the complementary set negligible as we have seen in Section 3.
Looking at
\begin{equation}
\begin{split}
\label{Tl>}
T_l^> =&- \frac{N-j}{\e}\sum_{i=1} ^j \int_0 ^t d\tau \int dX_{j+1}\int_{|v_i-v_{j+1}|> a\e^{1/4}} dV_{j+1} \int_0 ^\tau ds R_i(X_j,V_j,\tau)\.\\ 
& F\left(\frac{x_i-x_{j+1}}{\e}\right)\left\{F\left(\frac{x_i(-s)-x_{j+1}(-s)}{\e}\right)\.
[(\nabla_{v_k}-\nabla_{v_l})S(\tau-s)f_{j+1} ^0](X_{j+1},V_{j+1})\right\},
\end{split}
\end{equation}
we could apply the same argument as in Section 3 to get the result, if the motion of the pair of particles $i$ and $j+1$ would be independent of the others. However we have seen in the proof of Lemma 5 that the contribution of the event in which the particle $k\neq i,j+1$ interacts with particle $i$ or particle $j+1$ is indeed negligible.  Hence \eqref{result2} follows easily.  
\end{proof}
Finally, again by Lemma 1, we obtain \eqref{thm2gamma}.
\end{proof}

\noindent \textbf{Acknowledgements\\} 
We thank the referee for having pointed out useful suggestions to improve the exposition of a previous version of this paper.\\
The support from Swedish Research Council (VR grant 621-2009-5751) for A.V.B. is gratefully acknowledged.

%
% BibTeX users please use
% \bibliographystyle{}
% \bibliography{}
%
% Non-BibTeX users please use

\end{document}